\documentclass[11pt]{article}

\usepackage[margin=1in]{geometry}
\usepackage{amsmath,amssymb,amsthm}
\usepackage{mathtools}
\usepackage{hyperref}
\hypersetup{hidelinks}
\usepackage{cleveref}
\usepackage{xcolor}

\newtheorem{theorem}{Theorem}[section]
\newtheorem{lemma}[theorem]{Lemma}
\newtheorem{proposition}[theorem]{Proposition}
\newtheorem{claim}[theorem]{Claim}
\newtheorem{observation}[theorem]{Observation}
\newtheorem{corollary}[theorem]{Corollary}

\theoremstyle{definition}

\newtheorem{definition}[theorem]{Definition}
\newtheorem{remark}[theorem]{Remark}

\AddToHook{env/lemma/begin}{\crefalias{theorem}{lemma}}
\AddToHook{env/proposition/begin}{\crefalias{theorem}{proposition}}
\AddToHook{env/claim/begin}{\crefalias{theorem}{claim}}
\AddToHook{env/observation/begin}{\crefalias{theorem}{observation}}
\AddToHook{env/corollary/begin}{\crefalias{theorem}{corollary}}

\AddToHook{env/definition/begin}{\crefalias{theorem}{definition}}
\AddToHook{env/remark/begin}{\crefalias{theorem}{remark}}

\newcommand{\R}{\mathbb{R}}
\newcommand{\Z}{\mathbb{Z}}
\newcommand{\F}{\mathbb{F}}
\newcommand{\E}{\mathbb{E}}

\newcommand{\argmax}{\operatorname*{argmax}}

\newcommand{\cC}{{\mathcal C}}
\newcommand{\cD}{{\mathcal D}}

\newcommand{\cN}{{\mathcal N}}
\newcommand{\cP}{{\mathcal P}}
\newcommand{\cQ}{{\mathcal Q}}

\newcommand{\cE}{{\mathcal E}}
\newcommand{\cF}{{\mathcal F}}

\newcommand{\Tr}{\mathrm{Tr}}

\renewcommand{\epsilon}{\varepsilon}
\renewcommand{\phi}{\varphi}

\newcommand{\IE}{\mathbb{E}}
\newcommand{\IF}{\mathbb{F}}

\newcommand{\IR}{\mathbb{R}}

\newcommand{\IP}{\mathbb{P}}

\DeclareMathOperator{\poly}{poly}
\DeclareMathOperator{\rk}{rk}

\DeclareMathOperator{\supp}{supp}

\title{Universal set families for maximization of \\ nonnegative submodular and XOS functions}

\author{Chandra Chekuri \\ University of Illinois \\ Urbana-Champaign
\and Richard Ueltzen \\ Stanford University
\and Jan Vondr\'ak \\ Stanford University}

\begin{document}

\maketitle

\begin{abstract}
We consider the question of designing a universal family of sets $\cF \subset 2^{[n]}$ 
such that for any function $f:2^{[n]} \to \R_{\geq 0}$ in a certain class, we have
$$\max_{S \in \cF} f(S) \geq c(n) \cdot \max_{S \subset [n]} f(S).$$

We prove that there is a family of subpolynomial size such that for any nonnegative submodular function, $c(n) = \Omega(\frac{\log \log n}{\log n})$, and there is a family of logarithmic size such that $c(n) = \Omega(\frac{1}{\log n})$. We also prove that pairwise independence (which achieves a constant factor for graph cut functions), or even $k$-wise independence, does not imply a bound better than $O(\frac{1}{\sqrt{\log n}})$ for submodular functions. On the other hand, we prove that for any polynomially representable subclass of nonnegative submodular functions (such as the matroid connectivity functions for matroid representable over $\F_q$), a constant-factor universal family of polynomial size always exists.

For absolute XOS functions (a class that we introduce, in the form $f(S) = \max_i |\sum_{j \in S} w_{ij} + c_i|$ where $w_{ij}, c_i \in \R$), we design a family of polynomial size such that $c(n) \geq \sqrt{\frac{\log n}{n}}$, and prove that there is no polynomial-size family achieving a factor better than $O(\sqrt{\frac{\log n}{n}})$. 
\end{abstract}

\section{Introduction}

Discrete optimization in general aims to optimize an objective function $f(S)$ over a discrete set of solutions $S \in \cF$. The set of solutions often has a non-trivial combinatorial structure, but sometimes the complexity of the problem is completely encapsulated in the objective function $f$. A classical example is the Max Cut problem, where $\cF = 2^V$, the set of all subsets of vertices of a graph $G$, and the objective function is $f(S) = |E(S,V \setminus S)|$, the number of edges crossing the partition between $S$ and $V \setminus S$. 
This is the type of discrete optimization problem we study in this paper: optimization of a non-trivial set function $f:2^V \to \R_{\geq 0}$ over all subsets $S \subset V$. Beyond specific combinatorial problems, attention has focused on classes of objective functions with certain algebraic properties. A natural class (capturing for example the Max Cut problem, in both undirected and directed graphs), is the class of {\em submodular functions}.

\begin{definition}
$f:2^V \to \R$ is submodular if $f(S \cup T) + f(S \cap T) \leq f(S) + f(T)$ for all $S,T \subset V$.
\end{definition}

An alternative definition of submodularity is diminishing marginal utilities: $f(A \cup \{v\}) - f(A) \ge f(B \cup \{v\})-f(B)$ for all $A \subset B$ and $v \not \in B$. A function is \emph{monotone} if $f(A) \le f(B)$ for all $A \subset B$ and normalized if $f(\emptyset) = 0$. We use the term nonmonotone if the function is not necessarily monotone. $f$ is symmetric if $f(A) = f(V \setminus A)$ for all $A \subset V$. 

In the following, we always assume that $f$ is nonnegative.
We assume the most basic model of access to $f$ which is through a value oracle, returning $f(S)$ for a given set $S$. By the seminal results of \cite{GLS,Schrijver00,IwataFF01}, there is a polynomial-time algorithm for unconstrained submodular function
{\em minimization}. In contrast, the unconstrained submodular maximization problem is NP-hard (since it captures the Max Cut problem). It admits a $1/2$-approximation algorithm which is best possible in the value oracle model \cite{BuchbinderFNS15,FeigeMV07}. When $f$ is monotone the unconstrained maximization problem is trivial since $V$ is the optimum solution. The simplest non-trivial constrained optimization problem is to maximize a monotone $f$ subject to a cardinality constraint: $\max \{ f(S): |S| \leq k \}$.
A classical result \cite{NWF78,NW78} shows that a simple greedy algorithm yields a $(1-1/e)$-approximation for this problem and moreover this is tight \cite{NWF78,Feige98}. Maximization of submodular functions subject to various constraints has been subject to extensive research and  we do not attempt to survey here. 

Beyond the approximability of various submodular optimization problems, the following question was posed in \cite{BalkanskiS18}: Is it possible to approximate the optimal solution if we can access the objective function $f(S)$ in a limited way, in rounds of parallel queries $\cQ_1, \cQ_2, \ldots$ where each round of queries is executed in parallel and depends only on the answers to previous rounds of queries. The number of rounds should be small, ideally a constant or logarithmic in the size of the input. The term \emph{adaptivity} was coined to understand this round complexity.
The motivation comes from applications in which queries to $f(S)$ are the main bottleneck, due to necessary experiments or information gathering which is slow but can be executed efficiently in parallel. One can also see that the notion of adaptivity has strong connections to parallelization.
This model was initially considered for the cardinality-constrained problem $\max \{f(S): |S| \leq k \}$, with $f$ monotone submodular, and \cite{BalkanskiRS19a} showed that one can achieve a $(1-1/e-\epsilon)$-approximation in $O(\log n \log k \poly(1/\epsilon))$ rounds of queries. \cite{BalkanskiS18} also showed that a logarithmic number of rounds is necessary to achieve any constant factor and the lower bound was refined in \cite{li2020polynomial}. Several results on adaptivity for other constraints were obtained subsequently \cite{balkanski2019optimal,chekuri2019parallelizing,ene2019submodular}.

In the unconstrained (non-monotone) setting, it was shown that a constant number of rounds suffices to obtain a constant factor: more precisely, one can achieve a $(1/2 - \epsilon)$-approximation in $O(1/\epsilon)$ rounds of queries \cite{ChenFK19}. This leads to a natural question: \emph{Can one achieve a constant approximation without any adaptivity at all?}

In some sense, this question has a positive answer, due to the following well-known lemma:
For a uniformly random $R \subset V$, $\E[f(R)] \geq \frac14 \max_{S \subset V} f(S)$ \cite{FeigeMV07}.
Hence, in expectation, a random set already gives a constant-factor approximation, and by sampling a polynomial number of sets, we can find a solution of good value with high probability. However, we do not have a certificate that our algorithm succeeded and there is no way to verify the quality of the solution if we cannot ask further queries. 
Hence, the question we ask in this paper\footnote{(and this question has been posed in workshops on submodular optimization over the years)} is:

\medskip
{\em Is there a deterministic non-adaptive algorithm, i.e. an algorithm querying only one universal (predetermined) family of polynomially many sets, which achieves a constant-factor times the maximum for every nonnegative submodular function?}
\medskip

\subsection{Our results}

We prove the first non-trivial results on the question of universal approximating families, for nonnegative submodular functions and absolute XOS functions (which we define in this paper).

\paragraph{Universal families for submodular functions.}
For submodular functions, we prove the following.

\begin{theorem}\label{thm:submod-log}
For every $n \geq 2$, there is a family $\cF_n$ of size $|\cF_n| \leq 2 \lceil \log_2 n \rceil$ such that
for every submodular function $f:2^{[n]} \to \R_{\geq 0}$,
$$ \max_{S \in \cF_n} f(S) \geq \frac{1}{2 \lceil \log_2 n \rceil} \max_{S \subset [n]} f(S).$$
\end{theorem}

Assuming $n = 2^d$ and viewing $[n]$ as $\{0,1\}^d$, the specific family $\cF_n$ attaining this result is the family of coordinate-wise layers $\{0,1\}^d$.
With families of polynomial size, we can achieve a slightly improved result.

\begin{theorem}
\label{thm:submod-loglog}
For every $n \geq 2$, there is a family $\cF_n$ of size $|\cF_n| = n^{o(1)}$ such that
for every submodular function $f:2^{[n]} \to \R_{\geq 0}$,
$$ \max_{S \in \cF_n} f(S) \geq \frac{\log_2 \log_2 n}{10 \log_2 n} \max_{S \subset [n]} f(S).$$
\end{theorem}

The family here is obtained by expressing $n$ as $r^d$ and considering unions of parallel coordinate-wise layers in $[r]^d$.
The choice of $r \simeq \log n$ and $d \simeq \frac{\log n}{\log \log n}$ leads to our result. We remark that the family can be made as small as $|\cF_n| \simeq 2^{\log_2^\alpha n}$ for any $\alpha \in (0, \log_2 n / \log_2 \log_2 n]$, for an approximation factor of $\frac{\alpha \log_2 \log_2 n}{8 \log_2 n}$.
Note that the trade-off is not very sensitive to $|\cF_n|$, with asymptotically the same factor for all $\alpha > 0$.

Both of these results can be viewed as constructions of set families $\cF$ such that each pair of elements $i \neq j$ is separated by some set $S \in \cF$. Considering the Max Cut problem, this is a necessary condition for achieving any approximation of the maximum for submodular functions. 
By the symmetry of the problem, it is a natural idea to consider families supporting a {\em uniformly pairwise independent distribution}.

\begin{definition}
\label{def:pairwise-indep}
A distribution over a family $\cF \subset 2^{[n]}$ with probabilities $p_S \geq 0$, $\sum_{S \in \cF} p_S = 1$, is called uniformly pairwise independent, if for each $i \neq j \in [n]$ and $L \subset \{i, j\}$, we have 
$$\sum_{S \in \cF: \{i, j\} \cap S = L} p_S = \frac14.$$
\end{definition}

Uniformly pairwise independent distributions exist on a support of linear size (e.g. one given by the Hadamard matrix), and it is a folklore fact that a uniformly pairwise independent distribution is sufficient to extract a factor of $\frac12$ for the Max Cut problem, and $\frac14$ for Directed Max Cut. 
Perhaps this is true even for submodular functions in general? In other words, is the expectation of a submodular function over a uniformly pairwise independent distribution always within a constant factor of the expectation over the uniform distribution (which is within a constant factor of the maximum)?
Unfortunately, we prove that this is not the case.

\begin{theorem}
\label{thm:pairwise-indep}
For infinitely many $n>1$, there is a uniformly pairwise independent distribution $\cD$ over $2^{[n]}$ 
supported on $\cF$
and a symmetric submodular function $f:2^{[n]} \to \R_{\geq 0}$ such that
$$ \max_{S \in \cF} f(S) \leq O\left(\frac{1}{\sqrt{\log n}}\right) \max_{S \subset [n]} f(S).$$
\end{theorem}

The construction is algebraic, based on binary matroids related to Reed-Muller codes -- with a representation defined by evaluations of monomials over $\Z_2^d$.
We note that the arising submodular functions {\em polynomially representable} -- encoded by a matrix over $\F_2$. We prove that for any such subclass of submodular function, a constant factor can be actually achieved, by a simple union-bound argument. Nevertheless, the respective universal family is not explicit.

\begin{theorem}
\label{thm:poly-representable}
For any polynomially representable subclass $\cC$ of nonnegative submodular functions and any $\epsilon>0$, there is a universal family for each $n$ of size $|\cF_n| = poly(n)$ such that for every $f \in \cC_n$,
$$ \max_{S \in \cF_n} f(S) \geq \left(\frac{1}{4} - \epsilon \right) \max_{A \subset [n]} f(A).$$
\end{theorem}

\paragraph{Universal families for absolute XOS functions.}

Here, we consider the following classes. These are possibly non-monotone variants of the more commonly studied monotone XOS functions, which play a significant role in algorithmic game theory.

\begin{definition}
\label{def:XOS}
$f:2^V \to \R$ is a absolute XOS function, if there is a finite index set $I$ and
weight functions $w_i: V \rightarrow \R$ and $c_i \in \R$ such that 
$f(S) = \max_{i \in I} \left| \sum_{j \in S} w_{i}(j) + c_i\right|$.
\end{definition}

\begin{definition}
\label{def:XOS-connectivity}
$f:2^V \to \R$ is an XOS connectivity function, if there is a finite index set $I$ and
weight functions $w_i: V \rightarrow \R$ such that $w_i(V) = 0$ and 
$f(S) = \max_{i \in I} \left| \sum_{j \in S} w_{i}(j)\right|$.
\end{definition}

We remark that XOS connectivity functions are symmetric in the sense that $f(S) = f(V \setminus S)$, and they subsume submodular connectivity functions.
Absolute XOS functions are more general, but in fact do not subsume all nonnegative submodular functions. Nevertheless, a universal family for absolute XOS functions also implies a universal family for nonnegative submodular functions, via submodular connectivity functions (while losing a factor of $2$). See Appendix~\ref{sec:XOS-submod} for more details.

We prove the following.

\begin{theorem}
\label{thm:XOS-upper}
For every $n \geq 1$, there is a family $\cF_n$ of size $|\cF_n| = \poly(n)$ such that
for every absolute XOS function $f:2^{[n]} \to \R_{\geq 0}$,
$$ \max_{S \in \cF_n} f(S) \geq \Omega\left(\sqrt{\frac{\log n}{n}}\right) \max_{S \subset [n]} f(S).$$
\end{theorem}

This result is optimal, due to the following.

\begin{theorem}
\label{thm:XOS-lower}
For every sufficiently large $n \geq 1$, and every family $\cF \subset 2^{[n]}$, $|\cF| = \poly(n)$, 
there is a nonzero absolute XOS function $f:2^{[n]} \to \R_{\geq 0}$ such that
$$ \max_{S \in \cF} f(S) \leq O\left(\sqrt{\frac{\log n}{n}} \right) \max_{S \subset [n]} f(S).$$
\end{theorem}

More precisely, our upper and lower bounds determine the approximation factor $c_D(n)$ when $|\cF_n| = \Theta(n^D)$ up to an absolute multiplicative constant:
First, if $\cF_n$ has size at most $n$, then we can construct a nonzero affine linear function that vanishes on $\cF_n$, and so $c_D(n)$ is $0$ (eventually) when $D < 1$.
For $D \geq 1$, our bounds give $c_D(n) = \Theta\left(\sqrt{\frac{1 + (D-1) \log n}{n}}\right)$.

We remark that these results are related to combinatorial discrepancy: a universal family $\cF$
approximates the maximum of any XOS function, if there is no way to design signed weights
such that their balance on each set $\cF$ is small, but the total weight is large.
Hence a good universal family $\cF$ for XOS functions is a difficult instance for discrepancy, even with weighted assignments. 

\subsection{Other related work}

A related question which was studied is the question of {\em pointwise approximation} for submodular (and other) functions: approximation of $f$ by another function $g$ such that $f(S) \leq g(S) \leq \alpha f(S)$ for all $S \subset V$. If this can be done using a function $g$ in a suitable form,  this can be an avenue to approaching our problem. 

In particular, \cite{GoemansHIM09} proved that every monotone submodular function can be approximated pointwise within a factor of $\sqrt{n}$ by a function in the form $g(S) = \sqrt{w^T 1_S}$. \cite{BalcanHI12} extended this result to (non-monotone) symmetric submodular functions: They can be approximated within a factor of $\sqrt{n}$ by functions in the form $g(S) = \sqrt{1_S^T M 1_S}$ where $M$ is a positive definite matrix. Since the expectation of functions in this form can be captured well by a pairwise independent distribution, this approximation could be used to derive a universal family that approximates the maximum within a factor of $\sqrt{n}$ for symmetric submodular functions.

Note that this result is rather weak in the context of our results: For submodular functions, we obtain an $O(\frac{\log n}{\log \log n})$-approximating universal family. The reason is that pointwise approximation is much stronger than what we require of our universal family: to capture the maximum of $f(S)$ rather than its values on all sets. Nevertheless, this connection was an inspiration for our result for absolute XOS functions. The basis for the pointwise approximation results of \cite{GoemansHIM09,BalcanHI12} is the approximation of a centrally symmetric polytope by an ellipsoid. 
While the submodular structure of a polytope is important in the works of \cite{GoemansHIM09,BalcanHI12}, it is less crucial for our purposes. 
The class of functions related to general centrally symmetric polytopes is the class of absolute XOS functions such that $f(\emptyset) = 0$, i.e., in the form $f(S) = \max_{x \in P} |x(S)|$.
For this class, one can obtain a universal family approximating the maximum within a factor of $\sqrt{n}$ by several different arguments, one of them being the approximation of $P$ by an ellipsoid. We prove a slightly improved result here (Theorem~\ref{thm:XOS-upper}) by a different, more direct method.

\section{Universal sets for submodular functions}

For simplicity of exposition, we first reduce the problem to normalized \emph{symmetric} submodular functions, which we call ``submodular connectivity functions''.

\begin{definition}
A submodular connectivity function is a nonnegative symmetric submodular function $f$ such that $f(\emptyset)  = 0$.
\end{definition}

\subsection{Reduction to submodular connectivity functions}
\label{sec:reduction-to-submodular-connectivity-functions}

\begin{definition}
    Let $f : 2^V \to \IR$ be a submodular function and let $S \subset V$.
    We define $f^S_{\mathrm{sym}} : 2^S \to \IR$ by $f^S_{\mathrm{sym}}(A) = f(A) + f(S \setminus A) - f(S) - f(\emptyset)$.
\end{definition}

It is easy to check that $f^S_{\mathrm{sym}}$ is a submodular connectivity function on $S$:
$f^S_{\mathrm{sym}}$ is submodular since the maps $A \mapsto f(A)$ and $A \mapsto f(S \setminus A)$ are.
$f^S_{\mathrm{sym}}(A) = f^S_{\mathrm{sym}}(S \setminus A)$, $f^S_{\mathrm{sym}}(A) \geq 0$ by submodularity of $f$, and $f^S_{\mathrm{sym}}(\emptyset) = 0$.

In particular, for a general submodular function $f:2^V \to \R$, $f^V_{\mathrm{sym}}$ is a submodular connectivity function.
Moreover, assume $\cF$ is a universal approximating family for submodular connectivity functions, i.e.,
\begin{equation*}
    \max_{A \in \cF} f^V_{\mathrm{sym}}(A) \geq c(n) \max_{A \subset V} f^V_{\mathrm{sym}}(A)
\end{equation*}
for any choice of $f$.
Then we can augment $\cF$ by complements, to obtain $\cF' = \{ S, V \setminus S: S \in \cF \}$, and we get
\begin{align*}
    \max_{A \in \cF'} f(A)
    &\geq \frac12 \max_{A \in \cF} (f(A) + f(V \setminus A))\\
    &= \frac12 \max_{A \in \cF} (f^V_{\mathrm{sym}}(A) + f(\emptyset) + f(V))\\
    &\geq \frac12 c(n) \max_{A \subset V} (f^V_{\mathrm{sym}}(A) + f(\emptyset) + f(V))\\
    &\geq \frac12 c(n) \max_{A \subset V} f(A),
\end{align*}
where we used the definition of $c(n)$ and that $\frac{1}{2} c(n) \leq \frac{1}{2} \leq 1$ in the third step, and nonnegativity of $f$ in the fourth step.

Hence it is sufficient to find a universal family for submodular connectivity functions,
and this implies the same approximation within a factor of $2$ for general (nonnegative) submodular functions.

Let us prove now a basic property of submodular connectivity functions which will be important in the following.

\begin{lemma}\label{restriction-inequalities}
    For any submodular connectivity function $f : 2^V \to \IR$ and any $A, S \subset V$, 
    denote $S^c = V \setminus S$. Then we have
    \begin{equation*}
        f(A) - f(S) \leq \frac{f^S_{\mathrm{sym}}(A \cap S) + f^{S^c}_{\mathrm{sym}}(A \cap S^c)}{2} \leq f(A).
    \end{equation*}
\end{lemma}
\begin{proof}
Unravelling the definitions, we can write
\begin{equation*}
    f^S_{\mathrm{sym}}(A \cap S) + f^{S^c}_{\mathrm{sym}}(A \cap S^c) = f(A \cap S) + f(A \cap S^c) + f(A^c \cap S) + f(A^c \cap S^c) - f(S) - f(S^c).
\end{equation*}
Using the symmetry and submodularity of $f$,
    \begin{align*}
        2 f(A) - 2 f(S)
        &= f(A) + f(A^c) - f(S) - f(S^c)\\
        &\leq f(A \cap S) + f(A \cap S^c) + f(A^c \cap S) + f(A^c \cap S^c) - f(S) - f(S^c) \\
        &\leq f(A) + f(A^c)\\
        &= 2 f(A).
    \end{align*}
    Dividing the chain of inequalities by $2$ yields the desired claim.
\end{proof}

\subsection{A logarithmic universal family}
\label{sec:submod-log}

Let us now design a universal family for submodular connectivity functions. Our first construction will be very simple: Assuming that $n = 2^d$ (or extending the ground set by dummy elements if necessary), we can view the ground set as $V = \{0,1\}^d$. Then we define $\cF = \{ H_1,\ldots,H_d \}$ where $H_i = \{x \in \{0,1\}^d: x_i = 0 \}$. The idea here is that this is a minimal family such that each pair $i \neq j \in V$ is separated by at least one set in $\cF$. This is a necessary requirement to achieve any approximation of the maximum even for graph cut functions, because otherwise an edge between $\{i,j\}$ not separated by $\cF$ would cause the maximum to be positive while the value of each $S \in \cF$ is zero.

Hence, it is perhaps optimistic to hope that such a family approximates the maximum of any submodular function. But this turns out to be true, within a logarithmic approximation factor, for any such family.

\begin{definition}
    For $S_1,\ldots,S_d \subset V$, define $\sigma(S_1, \ldots, S_d)$ to be the family of all sets that can be generated from $S_1,\ldots,S_d$ by the operations of union, intersection, and complement within $V$.
\end{definition}

\begin{lemma}\label{covering-bound}
    Let $V$ be a finite set, let $f : 2^V \to \IR$ be a submodular connectivity function, and let $S_1, \dots, S_d$ be subsets of $V$.
    Then
    \begin{equation*}
        \max_{A \in \sigma(S_1, \dots, S_d)} f(A) \leq \sum_{i=1}^d f(S_i).
    \end{equation*}
\end{lemma}
\begin{proof}
    We prove the claim by induction on $d$.
    So, let $A \in \sigma(S_1, \dots, S_d)$.
    If $d = 0$, then $A \in \{\emptyset, V\}$, and then $f(A) = f(\emptyset) = 0 \leq 0$, as required.
    Now suppose $d \geq 1$ and that the claim holds for $d$ replaced with $d - 1$.
    Since $f^{S^d}_{\mathrm{sym}}$ is a submodular connectivity function, and $A \cap S_d \in \sigma_{S_d}(S_1 \cap S_d, \dots, S_{d-1} \cap S_d)$, the inductive assumption yields
    \begin{equation*}
        f^{S^d}_{\mathrm{sym}}(A \cap S_d) \leq \sum_{i=1}^{d-1} f^{S^d}_{\mathrm{sym}}(S_i \cap S_d).
    \end{equation*}
    Analogously,
    \begin{equation*}
        f^{\overline S_d}_{\mathrm{sym}}(A \cap \overline S_d) \leq \sum_{i=1}^{d-1} f^{\overline S_d}_{\mathrm{sym}}(S_i \cap \overline S_d).
    \end{equation*}
    Combining this with \Cref{restriction-inequalities}, we get
    \begin{align*}
        f(A)
        &\leq \frac{1}{2} f^{S^d}_{\mathrm{sym}}(A \cap S_d) + \frac{1}{2} f^{\overline S_d}_{\mathrm{sym}}(A \cap \overline S_d) + f(S_d)\\
        &\leq \sum_{i=1}^{d-1} \left(\frac{1}{2} f^{S^d}_{\mathrm{sym}}(S_i \cap S_d) + \frac{1}{2} f^{\overline S_d}_{\mathrm{sym}}(S_i \cap \overline S_d)\right) + f(S_d)\\
        &\leq \sum_{i=1}^{d-1} f(S_i) + f(S_d) = \sum_{i=1}^d f(S_i).\qedhere
    \end{align*}
\end{proof}

\begin{corollary}\label{submod-connectivity-log}
    Let $n \geq 2$ and let $V$ be a set of size $n$.
    Then there is a family $\cF$ of size $|\cF| \leq \lceil \log_2 n \rceil$ such that
    for every submodular connectivity function $f:2^V \to \R_{\geq 0}$,
    $$ \max_{S \in \cF} f(S) \geq \frac{1}{\lceil \log_2 n \rceil} \max_{S \subset V} f(S).$$
\end{corollary}
\begin{proof}
    We write $d \coloneqq \lceil \log_2 n\rceil > 0$, and identify $V$ with a subset of $\{0, 1\}^d$.
    We let $S_i = V \cap \{x \in \{0, 1\}^d : x_i = 0\}$, so $\sigma(S_1, \dots, S_d) = 2^V$.
    Therefore, $\cF = \{S_1, \dots, S_d\}$ satisfies the required properties by \Cref{covering-bound}.
\end{proof}

The reduction from nonnegative submodular functions to submodular connectivity functions described in \Cref{sec:reduction-to-submodular-connectivity-functions} gives the statement of \Cref{thm:submod-log}.

We remark that the logarithmic bound in \Cref{submod-connectivity-log} is tight for the universal family used in the proof.
Consider the hypercube $V = \{0,1\}^d$ and the edge set $E = \{ (x,y): x,y \in V, \|x-y\|_1 = 1\}$.
That is, $G = (V,E)$ is the usual hypercube graph of dimension $d$. The graph cut function associated with $G$ has a maximum of $d 2^{d-1}$, by considering the bipartition of $V$ into vertices of even/odd $\ell_1$ norm. However, our family of axis-aligned hyperplanes only achieves cuts of size $2^{d-1}$, by cutting all the edges in the respective dimension.

\subsection{A subpolynomial universal family}

Here we present an improvement of the construction from the previous section, by considering a larger universal family. In Section~\ref{sec:submod-log}, we considered a family of sets (or cuts) such that each pair is separated by at least one cut. Instead, let's consider partitions (or the associated $\sigma$-algebras) $\cP_1,\ldots,\cP_d$ (into possibly more than two parts) such that each pair is separated by at least one partition. It turns out that Lemma~\ref{covering-bound} generalizes to partitions in the following way.

\begin{definition}
For $\sigma$-algebras $\cP_1,\ldots,\cP_d$, we denote by $\sigma(\cP_1,\ldots,\cP_d)$ the $\sigma$-algebra generated by $\cP_1,\ldots,\cP_d$.
By $\IE_{S \in \cP_i}$, we denote the expectation over a uniformly random set from $\cP_i$.
\end{definition}

We remark that a $\sigma$-algebra on a finite set $V$ is defined by a partition $V = T_1 \sqcup T_2 \sqcup \ldots \sqcup T_\ell$, and a uniformly random element of $\cP_1$ is obtained by selecting each $T_i$ independently with probability $1/2$ and taking the union.

\begin{lemma}\label{partition-refinment-bound}
    Let $V$ be a finite set, let $f : 2^V \to \IR$ be a submodular connectivity function, and let $\cP_1, \dots, \cP_d$ be $\sigma$-algebras of $V$.
    Then
    \begin{equation*}
        \max_{A \in \sigma(\cP_1, \dots, \cP_d)} f(A) \leq 2\sum_{i=1}^d \IE_{S \in \cP_i}[f(S)].
    \end{equation*}
\end{lemma}
\begin{proof}
    We proceed by strong induction on $\sum_i |\cP_i| + |V|$.
    Let $A \in \sigma(\cP_1, \dots, \cP_d)$.
    We need to prove that $f(A) \leq 2\sum_{i=1}^d \IE_{S \in \cP_i}[f(S)]$.
    If all the $\cP_i$ are trivial, then the claimed inequality is $0 \leq 0$, which is true.
    Now suppose otherwise.
    Then, without loss of generality, we get that $\cP_1 \neq \{\emptyset, V\}$.
    Let $T \in \cP_1 \setminus \{\emptyset, V\}$ be inclusion-minimal.
    Then for every $i$, the family $\cP_i^T \coloneqq \{P \cap T \mid P \in \cP_i\}$ is a $\sigma$-algebra on $T$ and $A \cap T \in \sigma(\cP_1^T, \dots, \cP_d^T)$, so by induction,
    \begin{equation*}
        f^T_{\mathrm{sym}}(A \cap T) \leq 2\sum_{i=1}^d \IE_{S' \in \cP_i^T}[f^T_{\mathrm{sym}}(S')].
    \end{equation*}
    Similarly for $T^c = V \setminus T$,
    \begin{equation*}
        f^{T^c}_{\mathrm{sym}}(A \cap T^c) \leq 2\sum_{i=1}^d \IE_{S' \in \cP_i^{T^c}} [f^{T^c}_{\mathrm{sym}}(S')].
    \end{equation*}
    Averaging the two inequalities and using \Cref{restriction-inequalities} yields
    \begin{equation}\label{intermediate-inductive-bound}
        f(A) - f(T) \leq \frac{f^T_{\mathrm{sym}}(A \cap T) + f^{T^c}_{\mathrm{sym}}(A \cap T^c)}{2} \leq \sum_{i=1}^d \left(\IE_{S' \in \cP_i^T} [f^T_{\mathrm{sym}}(S')] + \IE_{S' \in \cP_i^{T^c}} [f^{T^c}_{\mathrm{sym}}(S')] \right).
    \end{equation}
    For $S' \in \cP_i^T$, the number $|\{S \in \cP_i : S \cap T = S'\}| = \frac{|\cP_i|}{|\cP_i^T|}$ is independent of $S'$.
    Consequently, when $S$ is sampled uniformly at random from $\cP_i$, then $S \cap T$ has the uniform distribution on $\cP_i^T$, and similarly $S \cap T^c$ has the uniform distribution on $\cP_i^{T^c}$.
    Therefore, by \Cref{restriction-inequalities},
    \begin{equation*}
        \IE_{S' \in \cP_i^T} [f^T_{\mathrm{sym}}(S')] + \IE_{S' \in \cP_i^{T^c}} [f^{T^c}_{\mathrm{sym}}(S')] = \IE_{S \in \cP_i} [f^T_{\mathrm{sym}}(S \cap T) + f^{T^c}_{\mathrm{sym}}(S \cap T^c)] \leq 2 \IE_{S \in \cP_i} [f(S)].
    \end{equation*}
    For $i = 1$, since we chose $T$ to be minimal in $\cP_1 \setminus \{\emptyset, V\}$, we have $\cP_1^T = \{\emptyset, T\}$ and hence, $f^T_{\mathrm{sym}}$ is $0$ on $\cP_1^T$.
    Therefore, we can improve the previous inequality as follows:
    \begin{align*}
        \IE_{S' \in \cP_1^T} [f^T_{\mathrm{sym}}(S')] + \IE_{S' \in \cP_1^{T^c}} [f^{T^c}_{\mathrm{sym}}(S')]
        &= \IE_{S' \in \cP_1^{T^c}} [f^{T^c}_{\mathrm{sym}}(S')] \\
        &= \IE_{S \in \cP_1} [f^{T^c}_{\mathrm{sym}}(S \cap T^c)] \\
        &= \IE_{S \in \cP_1} [f(S \cap T^c) + f(T^c \setminus S) - f(T^c)] \\
        &= \IE_{S \in \cP_1} [f(S \setminus T) + f(S \cup T) - f(T)]\\
        &= 2 \IE_{S \in \cP_1} [f(S)] - f(T).
    \end{align*}
    We used the fact that $T$ is a block in the partition defining $\cP_1$, and so it appears in a random $S \in \cP_1$ with probability $1/2$. Adding it or removing it with probability $1/2$ does not change the uniform distribution over $\cP_1$.
    
    Substituting these bounds into \Cref{intermediate-inductive-bound}, we conclude
    \begin{equation*}
        f(A) - f(T)
        \leq 2 \IE_{S \in \cP_1} [f(S)] - f(T) + 2 \sum_{i=2}^d \IE_{S \in \cP_i} [f(S)]
        = 2 \sum_{i=1}^d \IE_{S \in \cP_i} [f(S)] - f(T).
    \end{equation*}
    Adding $f(T)$ to both sides yields the required bound.
\end{proof}

Our generalized construction starts from identifying $[n]$ with the cube $[r]^d$ (assuming that
$n = r^d$, or again adding dummy elements if necessary). The $\sigma$-algebra $\cP_i$ for each $i$ is generated by a partition into layers orthogonal to dimension $i$. Clearly, each pair is separated by at least one $\cP_i$. We will need to include more sets in our family than before (in fact all possible sets in each $\sigma$-algebra $\cP_i$), but the advantage of this construction is that we can work with a lower dimension $d$. The next theorem captures the approximation provided by this family.

\begin{theorem}\label{arbitrary-hypercube-covering-bound}
    Suppose that $V$ has cardinality $n$ with $2 \leq n \leq r^d$ for positive integers $n, r, d$.
    Then there exists a non-empty $\cF \subset 2^V$ of size at most $d 2^r$ with the following property:
    For every submodular connectivity function $f : 2^V \to \IR$,
    \begin{equation*}
        \max_{A \subset V} f(A) \leq 2 d \max_{S \in \cF} f(S).
    \end{equation*}
\end{theorem}
\begin{proof}
    By $n \leq r^d$, we may identify $V$ with a subset of $[r]^d$.
    For $1 \leq i \leq d$, we let $\cP_i$ be the $\sigma$-algebra on $V$ generated by the $i$th coordinate, i.e.,
    \begin{equation*}
        \cP_i \coloneqq \left\{\{x \in V : x_i \in L\} : L \subset [r] \right\}.
    \end{equation*}
    We note that every $\cP_i$ has cardinality at most $2^r$, and that $\sigma(\cP_1, \dots, \cP_d) = 2^V$.
    By \Cref{partition-refinment-bound},
    \begin{equation*}
        \max_{A \subset V} f(A) \leq 2\sum_{i=1}^d \IE_{S \in \cP_i} f(S) \leq 2d \max_{S \in \cF} f(S),
    \end{equation*}
    where $\cF \coloneqq \bigcup_{i=1}^d \cP_i \neq \emptyset$. Finally, we indeed have $|\cF| \leq \sum_{i=1}^d |\cP_i| \leq d 2^r$.
\end{proof}

\begin{corollary}
    Let $n \geq 3$ and let $V$ be a set of size $n$.
    Then for any $\alpha \in (0, \log_2 n / \log_2 \log_2 n]$, there exists a non-empty $\cF \subset 2^V$ of size at most 
    $\frac{4 \log_2 n}{\alpha \log_2 \log_2 n} 2^{(\log_2 n)^\alpha}$
    such that any submodular connectivity function $f : 2^V \to \IR$ satisfies
    \begin{equation*}
        \max_{S \in F} f(S) \geq \frac{\alpha \log_2 \log_2 n}{4 \log_2 n} \max_{A \subset V} f(A).
    \end{equation*}
\end{corollary}
\begin{proof}
    Consider $\alpha \in (0, \log_2 n / \log_2 \log_2 n]$ and
    let $r$ be the least integer satisfying $r \geq (\log_2(n))^\alpha > 1$ and let $d$ be the least integer satisfying $r^d \geq n$.
    Then $r \leq n$ by $(\log_2 n)^\alpha \leq 2^{\log_2 n} = n$.
    Hence,
    \begin{equation*}
        d \leq 1 + \log_r(n) \leq 2 \frac{\log_2 n}{\log_2 r} \leq \frac{2 \log_2 n}{\alpha \log_2 \log_2 n}.
    \end{equation*}
    By \Cref{arbitrary-hypercube-covering-bound}, there exists a non-empty $\cF \subset 2^V$ of size at most $d 2^r \leq \frac{4 \log_2 n}{\alpha \log_2 \log_2 n} 2^{(\log_2 n)^\alpha}$ such that any submodular connectivity function $f : 2^V \to \IR$ satisfies
    \begin{equation*}
        \max_{A \subset V} f(A) \leq \frac{4 \log_2 n}{\alpha \log_2 \log_2 n} \max_{S \in \cF} f(S).\qedhere
    \end{equation*}
\end{proof}
A choice of $\alpha = 4/5$ and the reduction from nonnegative submodular functions to submodular connectivity functions described in \Cref{sec:reduction-to-submodular-connectivity-functions} gives the statement of Theorem~\ref{thm:submod-loglog}, with  $|\cF| = n^{o(1)}$.
We remark that other choices of $\alpha>0$ give a broad range of possible sizes of $\cF$, from subpolynomial to quasipolynomial, with a trade-off for the approximation factor. However, the asymptotic behavior of $\Theta(\frac{\log n}{\log \log n})$ is unchanged except for the leading constant. Hence, increasing the family size has little benefit in this construction.
\begin{remark}
    If we choose $\alpha$ growing with $n$, then we can achieve a much better approximation, at the cost of a superpolynomially sized $\cF$.
    For example, a choice of $\alpha = \epsilon \log_2 n / \log_2 \log_2 n$ with $\epsilon \in (0, 1]$ gives a family of size at most $\frac{4}{\epsilon} 2^{n^\epsilon}$ and an approximation factor of $\frac{\epsilon}{4}$.
\end{remark}

\subsection{A lower bound for pairwise independent distributions}

Note that the only property we use in our proofs so far is that a given family of sets or partitions separates every pair of elements. Moreover, one might conjecture that the weakness in our bounds and the reason they are non-constant is that different pairs are not separated uniformly often. In particular, in the $\{0,1\}^d$ hypercube example, our family of $d$ sets separates nearest neighbors only once, but antipodal pairs $d$ times. Hence it is natural to ask whether a {\em uniformly pairwise independent distribution} over sets (see Definition~\ref{def:pairwise-indep}) would achieve a constant approximation factor. Unfortunately, this turns out to be false, as the following (quantitative) restatement of \Cref{thm:pairwise-indep} shows.

\begin{theorem}
    For infinitely many positive integers $n$ there exists a family $\cF \subset 2^V$, where $|V| = n$ and with the following properties:
    The uniform distribution on $\cF$ is pairwise independent and there exists a nonzero submodular connectivity function $f : 2^V \to \IR$ such that
    \begin{equation*}
        \max_{H \in \cF} f(H) \leq \sqrt{\frac{2}{\pi \log_2(n / 2)}} \max_{A \subset V} f(A).
    \end{equation*}
\end{theorem}
\begin{proof}
    We claim that $n = 2^d$ works for any positive integer $m$ and $d = 2m + 1$.
    With these parameters, set $V = \IF_2^d$ and let
    \begin{equation*}
        \cF = \{\text{ affine hyperplanes in } V\} \cup \{\emptyset, V\} = \{H_{a, b} : a \in \IF_2^{d}, b \in \IF_2\},
    \end{equation*}
    where $H_{a, b} = \{x \in \IF_2^d : a \cdot x = b\}$.
    Then $\cF$ has cardinality $2n$ by injectivity of $(a, b) \mapsto H_{a, b}$.

    We next show that $\cF$ supports a pairwise independent distribution.
    For $x \in \IF_2^d$ and $H = H_{a, b}$ where $(a,b)$ is uniformly random in $\F \times \F$, we have
    \begin{equation*}
        \IP(x \in H \mid a) = \frac{1}{2}
    \end{equation*}
    since $a$ and $x \in H$ uniquely determine $b$.
    Similarly for the same random $H$, when $x, y \in \IF_2^d$ are distinct, then $x_i \neq y_i$ for some $i$, and hence for $L \subset \{x, y\}$,
    \begin{equation*}
        \IP(\{x, y\} \cap H = L \mid a_{[d] \setminus \{i\}}) = \frac{1}{4}
    \end{equation*}
    since $a_{[d] \setminus \{i\}}$ and $\{x, y\} \cap H = L$ uniquely determine $a_i$ and $b$.
    Taking expectations shows the uniform distribution on $\cF$ is pairwise independent.
    
    Let $P$ be the set of multilinear monomials of degree at most $m$ over $\IF_2$ in $d$ variables, so
    \begin{equation*}
        P = \{x_T = \prod_{j \in T} x_j \mid T \subset [d], |T| \leq m\} \subset \IF_2[x_1, \dots, x_d].
    \end{equation*}
    Then $P$ has size $\sum_{i = 0}^m \binom{d}{i} = n / 2$.
    The following observation follows by induction on $d$:
    \begin{observation}\label{P-lin-indep}
        Any collection of distinct multilinear monomials are linearly independent as functions $\IF_2^d \to \IF_2$.
    \end{observation}
    Therefore, the matrix $M \in \IF_2^{P \times V}$ with $M_{p, x} = p(x)$ has $|P| = n / 2$ rows, $n = |\IF_2^d|$ columns, and full rank $\rk(M) = |P| = n / 2$.
    Then $\rho : 2^V \to \IR$, $\rho(A) = \rk(M_A)$ with $M_A = M\mid_{P \times A}$ is the rank function of the linear matroid represented by $M$. It is a standard fact that the rank function of any matroid is monotone submodular (see Sections 39.1, 39.4 and 39.7 in \cite{schrijver-comb-opt}).
    Observe also that $\rho(\emptyset) = 0$ and $\rho(V) = n / 2$.

    Let us define $f:2^V \to \R$,
    $$ f(S) = \rho(S) + \rho(V \setminus S) - \rho(V). $$
    Then $f$ is symmetric submodular and $f(\emptyset) = f(V) = 0$, hence $f$ is a submodular connectivity function. 
   
    \begin{claim}
        $M M^T = 0$.
    \end{claim}
    \begin{proof}
        The entry indexed by $(x_T, x_{T'})$ of $M M^T \in \IF_2^{P \times P}$ equals $\sum_{x \in \IF_2^d} x_T x_{T'} = \sum_{x \in \IF_2^d} x_{T \cup T'}$.
        By $|T \cup T'| \leq 2m < d$, we have $T \cup T' \subsetneq [d]$, which implies $\sum_{x \in \IF_2^d} x_T x_{T'} = 0$.
    \end{proof}
    Since $M$ has full rank $n / 2$, we can choose $A \subset \IF_2^d$ of size $n / 2$, so that the square matrix $M_A$ is invertible, hence $\rho(A) = n / 2$.
    We have $\det(M_A M_A^T) = (\det (M_A))^2 \neq 0$. On the other hand, we have
    $M_A M_A^T + M_{V \setminus A} M_{V \setminus A}^T = M M^T = 0$. Hence
    $(\det (M_{V \setminus A}))^2 = \det(M_{V \setminus A} M_{V \setminus A}^T) = \det(M_A M_A^T) \neq 0$ and $M_{V \setminus A}$ is also invertible, which means $\rho(A) = \rho(V \setminus A) = n/2$.
    This shows
    \begin{equation*}
        f(A) = \rho(A) + \rho(V \setminus A) - \rho(V) = n / 2.
    \end{equation*}
    
    We now determine $\rho(H)$ for $H \in \cF$.
    If $H \in \{\emptyset, V\}$, then $f(H) = 0$ since $f$ is a submodular connectivity function.
    Now consider the case that $H$ is an affine hyperplane in $\cF$.
    Since the span of $P$ is preserved under invertible affine linear transformations on $\IF_2^d$, we get $\rho(H) = \rho(H_{e_1, 0})$.
    So, it is enough to consider the case $H = H_{e_1, 0}$.
    The restriction of $P$ to $H_{e_1, 0}$ is the set of multilinear monomials of degree at most $m$ over $H_{e_1, 0} \cong \IF_2^{d-1}$.
    So another application of \Cref{P-lin-indep} yields
    \begin{equation*}
        \rho(H) = \sum_{i=0}^{m} \binom{d-1}{i},
    \end{equation*}
    and similarly $\rho(H^c) = \rho(H)$.
    Consequently, by $2m = d-1$,
    \begin{equation*}
        f(H) = 2\sum_{i=0}^{m} \binom{d-1}{i} - n / 2 = \binom{d-1}{m}.
    \end{equation*}
    We conclude that for every $H \in \cF$,
    \begin{equation*}
        f(H) \leq \binom{d-1}{m} \leq \sqrt{\frac{1}{\pi m}} 2^{d-1} = \sqrt{\frac{2}{\pi (\log_2(n / 2))}} \frac{n}{2} = \sqrt{\frac{2}{\pi \log_2(n / 2)}} f(A).\qedhere
    \end{equation*}
\end{proof}

\subsection{A constant-factor universal family for polynomially representable submodular functions}

Note that the lower-bound example in the previous section is a submodular function in the form $f(S) = \rho(S) + \rho(V \setminus S) - \rho(V)$, where $\rho$ is the rank function of a matroid representable over $\F_2$, i.e. a binary matroid. Every such function $f$ has a representation of polynomial size, in this case by a matrix over $\F_2$. We call any such class of functions polynomially representable.

\begin{definition}
For a class $\cC$ of set functions, denote by $\cC_n \subset \cC$ the functions in $\cC$ defined on a ground set of size $n$. We call a class $\cC$ of set functions polynomially representable, if for each $n$ there is an injective encoding $e: \cC_n \to \{0,1\}^m$ such that $m = poly(n)$.
\end{definition}

Clearly, the number of functions in any polynomially representable class is simply exponential, in the form $2^{poly(n)}$. 
One example is the class of matroid rank functions for matroids representable over a fixed finite field $\F$, and the associated matroid connectivity functions $f(S) = r(S) + r(V \setminus S) - r(V)$.
In fact, even the class of matroid rank functions for representable matroids (over all fields) is polynomially representable by Theorem 1.1 in \cite{there-are-few-representable-matroids}.
Submodular functions, or even general matroid rank functions, are not polynomially representable; the number of possible matroids on $[n]$ is doubly exponential \cite{schrijver-comb-opt}. 
For any polynomially representable subclass of submodular functions, there is a simple probabilistic argument showing that there is a universal family of polynomial size, achieving a constant factor. 

\begin{theorem}
\label{thm:poly-representable}
For any polynomially representable subclass $\cC$ of nonnegative submodular functions and any fixed $\epsilon>0$, there is a universal family for each $n$ of size $|\cF_n| = poly(n)$ such that for every $f \in \cC_n$,
$$ \max_{S \in \cF_n} f(S) \geq \left(\frac{1}{4} - \epsilon \right) \max_{A \subset [n]} f(A).$$
\end{theorem}

\begin{proof}
We can assume $\epsilon \in (0,\frac{1}{4})$.
For each $n$, let $\cF_n = \{ R_1,\ldots,R_r \}$ be a collection of $r(n)$ independently and uniformly sampled subsets of $[n]$. (We determine $r(n)$ later.) We use the following well-known bound \cite{FeigeMV07}: For each of our random sets, $\E[f(R_i)] \geq \frac14 M$ where $M = \max_{A \subset [n]} f(A)$. Since $f(R_i) \in [0,M]$, we also obtain $\IP[f(R_i) < (\frac14 - \epsilon) M] \leq 1-\epsilon$. (Otherwise $\E[f(R_i)] < (1-\epsilon) (\frac14 - \epsilon) M + \epsilon M < \frac14 M$.) We sample $r$ sets independently, and so
$$ \IP\left[\forall i \in [r]; f(R_i) < \left(\frac14 - \epsilon \right) M \right] \leq (1 - \epsilon)^r.$$
Since $\cC$ is polynomially representable, we have $|\cC_n| \leq 2^{n^C}$ for some constant $C>0$. By the union bound,
$$ \IP\left[\exists f \in \cC_n; \forall i \in [r]; f(R_i) < \left(\frac14 - \epsilon \right) M \right] \leq 2^{n^C} (1 - \epsilon)^r.$$
Let us choose $r = \lceil \frac{1}{\epsilon} n^C \rceil$; then the probability bound is at most $ 2^{n^C} e^{-n^C} \ll 1$. Hence, there exists a universal family $\cF_n$ which approximates the maximum within $\frac{1}{4} - \epsilon$ for all $f \in \cC_n$.
\end{proof}

\section{Universal families for absolute XOS functions}

Next, we move on to the more general class of absolute XOS functions.  Here we prove the upper and lower bounds announced in Theorems~\ref{thm:XOS-upper} and \ref{thm:XOS-lower}, which are matching up to a constant factor. We will need the following known results.

\paragraph{Anticoncentration.}

We will use the following anticoncentration result from \cite{rademacher-anticoncentration}.
\begin{theorem}\label{montgomery-smith}
    There exists a universal constant $C > 1$ such that for any nonzero $w \in \ell_2$ and any positive real $t$ and a independent, and uniformly random $x_1, x_2, \dots \in \{-1, 1\}$, it holds that
    \begin{equation*}
        \IP\left(\sum_{i=1}^\infty x_i w_i > C^{-1} K_{1, 2}(w, t)\right) \geq C^{-1} \exp(-C t^2),
    \end{equation*}
    where $K_{1, 2}(w, t) = \inf\{\|w'\|_{\ell_1} + t\|w''\|_{\ell_2} : w', w'' \in \ell_2, w' + w'' = w\}$ is a mixture of the $\ell_1$- and $\ell_2$-norms.
\end{theorem}
We will apply the theorem in a slightly different form that suits our purposes.
The statement below says that when $x \in \{-1, 1\}^n$ is sampled uniformly at random and for $w \in \IR^n$ with fixed $\ell_1$-norm, the weight of the moderate tails of the distribution of $x^T w$ are essentially lightest when $w$ is a constant vector.
For completeness, we now give the formal deduction of that claim.
\begin{theorem}\label{chernoff-converse}
    There exists a universal constant $C > 1$ such that for any nonzero $w \in \IR^n$ and every $0 < t \leq 1$, if $x$ is chosen uniformly from $x \in \{-1, 1\}^n$, it holds that
    \begin{equation*}
        \IP(x^T w > C^{-1} t \|w\|_{\ell_1}) \geq C^{-1} \exp(-C t^2 n).
    \end{equation*}
\end{theorem}
\begin{proof}
    We view $w \in \ell_2$ by setting $w_i = 0$ for all $i > n$, and we suppose $x_1, x_2, \dots \in \{-1, 1\}$ are independent, and uniformly random.
    Suppose $w', w'' \in \ell_2$ satisfy $w' + w'' = w$.
    Let $w_1', w_1''$ be obtained from $w', w''$ by setting all entries at indices $i > n$ to $0$.
    Then we still have $w_1' + w_1'' = w$ since $w_i = 0$ for all $i > n$.
    Furthermore, $\sqrt{n} \|w_1''\|_{\ell_2} \geq \|w_1''\|_{\ell_1}$ since $w_1''$ is supported on $[n]$ and by Cauchy-Schwarz.
    Therefore,
    \begin{equation*}
        \|w'\|_{\ell_1} + t \sqrt{n} \|w''\|_{\ell_2} \geq \|w'_1\|_{\ell_1} + t \sqrt{n} \|w''_1\|_{\ell_2} \geq \|w_1'\|_{\ell_1} + t\|w_1''\|_{\ell_1} \geq t\|w\|_{\ell_1}
    \end{equation*}
    by the triangle-inequality and $t \leq 1$.
    Therefore, $K_{1, 2}(w, t\sqrt{n}) \geq t \|w\|_{\ell_1}$.
    Hence, by \Cref{montgomery-smith},
    \begin{equation*}
        \IP\left(x^t w > C^{-1} t \|w\|_{\ell_1})\right) \geq \IP\left(\sum_{i=1}^\infty x_i w_i > C^{-1} K_{1, 2}(w, t\sqrt{n})\right) \geq C^{-1} \exp(-C t^2 n).\qedhere
    \end{equation*}
\end{proof}

\paragraph{Discrepancy.}
We will also use Spencer's famous discrepancy theorem from \cite{six-standard-deviations-spencer}.
\begin{theorem}\label{sperners-discrepancy-theorem}
    Let $n$ be sufficiently large and let $\mathcal{F}$ be a family of $m$ subsets of an $n$-element set $X$,
    where $m\geq n$.
    Then there exists a coloring
    $\chi\colon X\to\{-1,1\}$ such that
    \begin{equation*}
        \max_{S\in\mathcal{F}}
        \left|\sum_{i \in S}\chi(i)\right|
        \leq
        11\sqrt{n\ln\left(\frac{2m}{n}\right)}
        \leq 10 \sqrt{n\log_2\left(\frac{2m}{n}\right)}.
    \end{equation*}
\end{theorem}

\subsection{A linear universal family}

We first construct a universal family of size $n+1$ with approximation factor $c(n) = \Omega(1 / \sqrt{n})$ for absolute XOS functions.
The proof uses the existence of Hadamard bases, which we shall briefly demonstrate.

\begin{definition}
    We say $v_1, \dots, v_n \in \IR^n$ form a Hadamard basis if $v_1, \dots, v_n \in \{-1, 1\}^n$ and they form an orthogonal basis for $\IR^n$.
\end{definition}
Notice that $v_1 = 1 \in \IR^1$ forms a Hadamard basis.
Furthermore, a simple calculation shows that for a Hadamard basis $v_1, \dots, v_n \in \IR^n$, the vectors
\begin{equation*}
    w_i \coloneqq \begin{cases}
        (v_i, v_i) & 1 \leq i \leq n,\\
        (v_{i-n}, -v_{i-n}) & n+1 \leq i \leq 2n
    \end{cases}
    \in \{-1, 1\}^{2n}
\end{equation*}
form a Hadamard basis of $\IR^{2n}$.
Hence, we obtain the following standard fact.
\begin{corollary}\label{existence-of-hadamard-basis}
    $\IR^n$ admits a Hadamard basis $v_1, \dots, v_n$ with $v_1 = 1$ for any $n$ that is a power of $2$.
\end{corollary}

\begin{theorem}\label{symmetric-xos-approximation-by-hadamard-basis}
    Let $V$ be a set of size $n$ and suppose $n$ is a power of $2$.
    Then there exists a family $\cF \subset 2^V$ of size $n + 1$ with the following property:
    For any absolute XOS $f : 2^V \to \IR$,
    \begin{equation}\label{eq:symmetric-xos-approximation-by-hadamard-basis}
        \max_{S \in \cF} f(S) \geq \frac{1}{3\sqrt{n}}\max_{A \subset V} f(A).
    \end{equation}
\end{theorem}
\begin{proof}
    Since $n$ is a power of $2$, \Cref{existence-of-hadamard-basis} shows that $\IR^n$ admits a Hadamard basis $v_1, \dots, v_n$ with $v_1 = 1$.
    Define $S_j$ so that $v_j = 1_{S_j} - 1_{V \setminus S_j}$ for $1 \leq j \leq n$, so $S_1 = V$.
    Further, set $S_0 \coloneqq \emptyset$.
    We claim that $\cF = \{S_0, S_1, S_2, \dots, S_n\}$ works.
    Clearly, $\cF$ has size $n+1$.

    Let $f$ be absolute XOS, i.e., $f(A) = \max_{1 \leq i \leq m} |w_i(A) + c_i|$ for a finite collection of $(w_i, c_i) \in \IR^V \times \IR$.
    Let us consider any $A \subset V$.
    We need to show that
    \begin{equation*}
        f(A) \leq 3\sqrt{n} \max_{0 \leq j \leq n} f(S_j).
    \end{equation*}
    Pick $i^\ast \in \argmax_{1 \leq i \leq m} |w_i(A) + c_i|$.
    Since $v_1, \dots, v_n$ form an orthogonal basis for $\IR^n$, we can write $v \coloneqq 1_A - 1_{V \setminus A}$ as
    \begin{equation*}
        v = \sum_{j=1}^n \frac{\langle v, v_j\rangle}{\|v_j\|^2} v_j.
    \end{equation*}
    Therefore,
    \begin{align*}
        f(A)
        &= |w_{i^\ast}(A) + c_{i^\ast}|\\
        &= \left|\frac{w_{i^\ast}^T v + w_{i^\ast}(V)}{2} + c_{i^\ast}\right|\\
        &= \left|\sum_{j=1}^n \frac{\langle v, v_j\rangle}{\|v_j\|^2} \left(\frac{w_{i^\ast}^T v_j + w_{i^\ast}(V)}{2} + c_{i^\ast}\right)
        - \left(\sum_{j=1}^n \frac{\langle v, v_j\rangle}{\|v_j\|^2} - 1\right) \left(\frac{w_{i^\ast}(V)}{2} + c_{i^\ast}\right)\right|\\
        &\leq \sum_{j=1}^n \frac{|\langle v, v_j\rangle|}{\|v_j\|^2} f(S_j) + \left(1 + \sum_{j=1}^n \frac{|\langle v, v_j\rangle|}{\|v_j\|^2}\right) \frac{f(\emptyset) + f(V)}{2}\\
        &\leq 3\sqrt{n} \max_{0 \leq j \leq n} f(S_j),
    \end{align*}
    where we used that $\emptyset = S_0, V = S_1$, and
    \begin{equation*}
        \sum_{j=1}^n \frac{|\langle v, v_j\rangle|}{\|v_j\|^2} = \sum_{j=1}^n \left|\left\langle \frac{v}{\|v\|}, \frac{v_j}{\|v_j\|}\right\rangle\right| \leq \sqrt{\left(\sum_{j=1}^n \left|\left\langle \frac{v}{\|v\|}, \frac{v_j}{\|v_j\|}\right\rangle\right|^2\right) \left(\sum_{j=1}^n 1^2 \right)} = \sqrt{n}.\qedhere
    \end{equation*}
\end{proof}

\subsection{A polynomial universal family}

We next construct a universal family of size roughly $n^D$ (for $D > 1$) with approximation factor $c(n) = \Omega(\sqrt{(D-1)\log n / n})$ for absolute XOS functions.
The construction is randomized, and we use a net-argument to show that it satisfies the required properties with positive probability.
Plugging in, say, $D = 2$ below yields \Cref{thm:XOS-upper}.

\begin{theorem}\label{non-adaptive-poly-size-approximation-for-XOS}
    Let $D > 1$ be a fixed real number.
    For any sufficiently large $n$ and any set $V$ of size $n$, there exists a family $\cF \subset 2^V$ of size at most $n^D$ with the following property:
    For any absolute XOS $f : 2^V \to \IR$,
    \begin{equation}\label{target-inequality-non-adaptive-poly-size-approximation-for-XOS}
        \max_{S \in \cF} f(S) \geq \sqrt{\frac{D-1}{8C^3}}\sqrt{\frac{\log n}{n}} \max_{A \subset V} f(A),
    \end{equation}
    where $C > 1$ is the constant from \Cref{chernoff-converse}.
\end{theorem}
\begin{proof}
    Set $N = \lfloor n^D \rfloor$, and sample $S_1, \dots, S_N \subset V$ independently and uniformly at random.
    We will show that $\cF = \{S_1, \dots, S_N\}$ satisfies the required properties with positive probability.
    Clearly, $|\cF| \leq N \leq n^D$.

    It remains to prove that \Cref{target-inequality-non-adaptive-poly-size-approximation-for-XOS} holds for some outcome of $\cF$.
    Let us start with an overview of the proof.
    We will need to verify \Cref{target-inequality-non-adaptive-poly-size-approximation-for-XOS} only for the case that $f(A) = |w(A) + c|$ for $(w, c) \in \IR^V \times \IR$ is defined by the maximum of a single term.
    To achieve this, we use a net argument: We cast an $\epsilon$-net in $\IR^V \times \IR \cong \IR^{n+1}$ to approximate $(w, c)$, and proceed by a union bound.
    
    We first compute the maximum of the function $f$ as
    \begin{equation}\label{maximum-of-affine-linear-function}
        \max_{A \subset V} |w(A) + c| = \max_{x \in \{-1, 1\}^V} \left|\frac{w^T x + w(V)}{2} + c\right| = \frac{1}{2} \left(\|w\|_{\ell_1} + |w(V) + 2c|\right).
    \end{equation}
    Let us from now on use the abbreviation $p(w, c) \coloneqq \|w\|_{\ell_1} + |w(V) + 2c|$.
    We will construct the net in
    \begin{equation}\label{definition-of-S}
        S = \{(w, c) \in \IR^V \times \IR : p(w, c) = 1\}.
    \end{equation}
    Note that $(w, c) \in S$ implies $|2c| \leq 1 - \|w\|_{\ell_1} + |w(V)| \leq 1$, hence $\|(w, c)\|_{\ell_1} \leq \frac{3}{2} < 2$.
    This means $S \subset B_{\ell_1}(0, 2)$.
    We set
    \begin{equation*}
        \epsilon \coloneqq \sqrt{\frac{D - 1}{32 C^3}} \sqrt{\frac{\log n}{n}} > 0.
    \end{equation*}
    Let $\cN \subset S$ be a collection of $\epsilon$-separated points in $S$, i.e., any two distinct points in $\cN$ have $\ell_1$-distance at least $\epsilon$.
    Then the open balls $B_{\ell_1}((w, c), \epsilon / 2)$ are disjoint and contained in $B_{\ell_1}(S, 1) \subset B_{\ell_1}(0, 3)$ for $(w, c) \in S$.
    Denoting the Lebesgue measure in $\IR^V \times \IR$ by $\lambda$ and setting $v_{n+1} = \lambda(B_{\ell_1}(0, 1)) > 0$, we get
    \begin{equation*}
        |\cN| \left(\frac{\epsilon}{2}\right)^{n+1} v_{n+1} = \sum_{(w, c) \in \cN} \lambda\left(B_{\ell_1}\left((w, c), \frac{\epsilon}{2}\right)\right) \leq \lambda(B_{\ell_1}(0, 3)) = 3^{n+1} v_{n+1}.
    \end{equation*}
    Hence, $|\cN| \leq (6 / \epsilon)^{n+1} \leq \exp(n^{1 + o(1)})$.
    Let us from now also assume that $\cN \subset S$ is an $\epsilon$-net, i.e., it is a maximal $\epsilon$-separated subset of $S$.
    This set can, for example, be built by greedily adding any point from $S \setminus B_{\ell_1}(\cN, \epsilon)$ to $\cN$ as long as possible.
    
    Using the tight anticoncentration result \Cref{chernoff-converse}, we next prove:
    \begin{claim}
        For $n$ large enough, it holds for all $(w, c) \in S$ that
        \begin{equation*}
            \IP\left(\max_{1 \leq j \leq N} |w(S_j) + c| < 2 \epsilon \right) \leq \exp\left(-N C^{-1} n^{-\frac{2}{3}(D-1)}\right).
        \end{equation*}
    \end{claim}
    \begin{proof}
        Since $S_1, \dots, S_N$ are independent, identically distributed with $x_1 = 1_{S_1} - 1_{V \setminus S_1}$, we can rewrite the left-hand side as $\IP\left(\left|x_1^T w + w(V) + 2c\right| < 4\epsilon\right)^N$.
        If $|w(V) + 2c| \geq 4 \epsilon$, then
        \begin{equation*}
            \IP\left(\left|x_1^T w + w(V) + 2c\right| \geq 4\epsilon\right) \geq \frac{1}{2}
        \end{equation*}
        since the distribution of $x_1^T w$ is symmetric around $0$.
        Otherwise, we get $\|w\|_{\ell_1} \geq 1 - 4\epsilon$ since $(w, c) \in S$.
        Then we write
        \begin{equation*}
            \IP\left(\left|x_1^T w + w(V) + 2c\right| \geq 4\epsilon\right)
            \geq \min\left\{\IP\left(x_1^T w \geq 4\epsilon\right), \IP\left(x_1^T w \leq -4\epsilon\right)\right\}
            = \IP\left(x_1^T w \geq 4\epsilon\right),
        \end{equation*}
        and lower-bound the right-hand side using \Cref{chernoff-converse}.
        Suppose $n$ is large enough, so that $\epsilon < \frac{1}{100}$, and $t \leq 1$, where
        \begin{equation*}
            t \coloneqq \frac{4C \epsilon}{\|w\|_{\ell_1}} \leq \frac{1}{1 - 4\epsilon} \sqrt{\frac{D - 1}{2 C}} \sqrt{\frac{\log n}{n}}.
        \end{equation*}
        Then applying \Cref{chernoff-converse}, we get
        \begin{equation*}
            \IP\left(x_1^T w \geq 4 \epsilon\right)
            = \IP\left(x_1^T w
            \geq C^{-1} t \|w\|_{\ell_1}\right) \geq C^{-1} e^{-C t^2 n}
            \geq C^{-1} n^{-(D-1) / (2(1 - 4\epsilon)^2)}
            \geq C^{-1} n^{-\frac{2}{3}(D-1)}.
        \end{equation*}
        In both cases, we conclude that 
        \begin{equation*}
            \IP\left(\left|x_1^T w + w(V) + 2c\right| \geq 4\epsilon\right) \geq C^{-1} n^{-\frac{2}{3}(D-1)}.
        \end{equation*}
        for large enough $n$. Then
        \begin{equation*}
            \IP\left(\max_{1 \leq j \leq N} |w(S_j) + c| < 2 \epsilon \right) \leq \left(1 - C^{-1} n^{-\frac{2}{3}(D-1)}\right)^N \leq \exp\left(-N C^{-1} n^{-\frac{2}{3}(D-1)}\right).\qedhere
        \end{equation*}
    \end{proof}
    Let $\cE$ be the event that $\max_{1 \leq j \leq N} |w(S_j) + c| \geq 2\epsilon$ holds for all $(w, c) \in \cN$.
    Then $\cE$ holds with probability
    \begin{equation*}
        \IP(\cE) \geq 1 - |\cN| \exp\left(-N C^{-1} n^{-\frac{2}{3}(D-1)}\right) \geq 1 - \exp\left(n^{1 + o(1)} - N C^{-1} n^{-\frac{2}{3}(D-1)}\right) > 0
    \end{equation*}
    for sufficiently large $n$ since $N C^{-1} n^{-\frac{2}{3}(D-1)} \geq n^{1 + \Omega(1)}$.
    Then we can choose an outcome for $\cF = \{S_1, \dots, S_N\}$ such that $\cE$ holds.
    We will prove that this $\cF$ satisfies \Cref{target-inequality-non-adaptive-poly-size-approximation-for-XOS}.
    
    Let $f : 2^V \to \IR$ be absolute XOS, i.e., $f(A) = \max_i |w_i(A) + c_i|$ for a finite collection of $(w_i, c_i) \in \IR^V \times \IR$.
    Then by \Cref{maximum-of-affine-linear-function},
    \begin{equation*}
        \max_{A \subset V} f(A) = \max_i \max_{A \subset V} |w_i(A) + c_i| = \frac{1}{2} \max_i p(w_i, c_i).
    \end{equation*}
    Since the conclusion of the theorem is otherwise trivial, we may assume that the left-hand side is strictly positive.
    Then $i^\ast$ maximizing the right-hand side satisfies $(w_{i^\ast}, c_{i^\ast}) \neq 0$.
    Since $\cN \subset S$ is an $\epsilon$-net, there exists a $(w, c) \in \cN$ with $\left\|(w, c) - \frac{(w_{i^\ast}, c_{i^\ast})}{p(w_{i^\ast}, c_{i^\ast})}\right\|_{\ell_1} < \epsilon$.
    Applying $\cE$ with $(w, c) \in \cN$, we get that there exists an $S_j \in \cF$ such that $|w(S_j) + c| \geq 2\epsilon$.
    So,
    \begin{align*}
        f(S_j)
        &\geq |w_{i^\ast}(S_j) + c_{i^\ast}|\\
        &\geq p(w_{i^\ast}, c_{i^\ast}) \left(|w(S_j) + c| - \left\|(w, c) - \frac{(w_{i^\ast}, c_{i^\ast})}{p(w_{i^\ast}, c_{i^\ast})}\right\|_{\ell_1}\right)\\
        &\geq p(w_{i^\ast}, c_{i^\ast}) (2\epsilon - \epsilon)\\
        &= 2\epsilon \max_{A \subset V} f(A).
    \end{align*}
    This shows
    \begin{equation*}
        \max_{S \in \cF} f(S) \geq 2\epsilon \max_{A \subset V} f(A) = \sqrt{\frac{D-1}{8C^3}}\sqrt{\frac{\log n}{n}} \max_{A \subset V} f(A).\qedhere
    \end{equation*}
\end{proof}

\subsection{A matching lower bound for our universal families in all regimes}

The following negative result directly yields \Cref{thm:XOS-lower}.
Substituting $D = 1$ (resp. $D > 1$) shows that the approximation factor in \Cref{symmetric-xos-approximation-by-hadamard-basis} (resp. \Cref{non-adaptive-poly-size-approximation-for-XOS} with parameter $D > 1$) is optimal up to a constant factor.
It is a direct consequence of Spencer's discrepancy theorem (\Cref{sperners-discrepancy-theorem}) from \cite{six-standard-deviations-spencer}.
We include the proof for the convenience of the reader.
\begin{proposition}\label{xos-quantitative-lower-bound}
    Let $n$ be large enough, let $V$ be a finite set of cardinality $n$.
    Then for any real numbers $K, D \geq 1$ and $\cF \subset 2^V$ of size $n \leq |\cF| \leq K n^D$, there exists a nonzero absolute XOS $f : 2^V \to \IR$ with
    \begin{equation*}
        \max_{S \in \cF} f(S) \leq 20 \sqrt{\frac{\log_2(2 K) + (D-1)\log_2(n)}{n}} \max_{A \subset V} f(A).
    \end{equation*}
\end{proposition}
\begin{proof}
    By \Cref{sperners-discrepancy-theorem}, there exists a coloring $\chi : V \to \{-1, 1\}$ such that
    \begin{equation*}
        \max_{S\in\mathcal{F}}
        \left|\sum_{i \in S}\chi(i)\right|
        \leq
        10 \sqrt{n\log_2\left(\frac{2|\cF|}{n}\right)}.
    \end{equation*}
    Then set $f : 2^V \to \IR$, $f(A) \coloneqq \left|\sum_{i \in A} \chi(i) \right|$.
    Then $f$ is nonzero absolute XOS with
    \begin{equation*}
        \max_{S \in \cF} f(S) \leq 
        10 \sqrt{n\log_2\left(\frac{2|\cF|}{n}\right)}.
    \end{equation*}
    On the other hand, we can bound
    \begin{align*}
        \max_{A \subset V} f(A)
        &= \max\{f(\{i : \chi(i) = 1\}), f(\{i : \chi(i) = -1\})\}\\
        &\geq \frac{n}{2} \\
        &\geq \frac{1}{20} \sqrt{\frac{n}{\log_2\left(\frac{2|\cF|}{n}\right)}} \max_{S \in \cF} f(S)\\
        &\geq \frac{1}{20} \sqrt{\frac{n}{\log_2(2 K) + (D-1)\log_2(n)}} \max_{S \in \cF} f(S).\qedhere
    \end{align*}
\end{proof}

\section{Conclusions}

The main remaining question is whether there is a polynomial-size family $\cF_n$ such that $$\max_{S \in \cF_n} f(S) \geq c \cdot \max_{S \subset V} f(S)$$ for a constant $c>0$ and all nonnegative submodular functions $f$. We know that the support of a pairwise-independent or even a $k$-wise independent distribution is not sufficient for this. It also remains to be resolved whether pairwise-independent distributions approximate the maximum of a submodular function within a factor of $O(\sqrt{\log n})$ or whether a stronger lower bound exists.

\paragraph{AI use statement:}
ChatGPT 5.6 Sol Pro was used to find the inductive proofs of Theorems~\ref{thm:submod-log} and \ref{thm:submod-loglog}, as well as the lower bound examples for Theorem~\ref{thm:pairwise-indep} and Theorem~\ref{thm:k-wise}.
We verified and digested the proofs, and wrote our own exposition of the results.

\bibliographystyle{amsplain}
\bibliography{refs}

@article {rademacher-anticoncentration,
    AUTHOR = {Montgomery-Smith, S. J.},
     TITLE = {The distribution of {R}ademacher sums},
   JOURNAL = {Proc. Amer. Math. Soc.},
  FJOURNAL = {Proceedings of the American Mathematical Society},
    VOLUME = {109},
      YEAR = {1990},
    NUMBER = {2},
     PAGES = {517--522},
      ISSN = {0002-9939,1088-6826},
   MRCLASS = {60C05 (46M35 60E15 60G50)},
  MRNUMBER = {1013975},
MRREVIEWER = {I.\ Ya.\ Novikov},
       DOI = {10.2307/2048015},
       URL = {https://doi.org/10.2307/2048015},
}

@article {six-standard-deviations-spencer,
    AUTHOR = {Spencer, Joel},
     TITLE = {Six standard deviations suffice},
   JOURNAL = {Trans. Amer. Math. Soc.},
  FJOURNAL = {Transactions of the American Mathematical Society},
    VOLUME = {289},
      YEAR = {1985},
    NUMBER = {2},
     PAGES = {679--706},
      ISSN = {0002-9947,1088-6850},
   MRCLASS = {05A05},
  MRNUMBER = {784009},
       DOI = {10.2307/2000258},
       URL = {https://doi.org/10.2307/2000258},
}

@book {schrijver-comb-opt,
    AUTHOR = {Schrijver, Alexander},
     TITLE = {Combinatorial optimization. {P}olyhedra and efficiency. {V}ol.
              {B}},
    SERIES = {Algorithms and Combinatorics},
    VOLUME = {24,B},
      NOTE = {Matroids, trees, stable sets,
              Chapters 39--69},
 PUBLISHER = {Springer-Verlag, Berlin},
      YEAR = {2003},
     PAGES = {i--xxxiv and 649--1217},
      ISBN = {3-540-44389-4},
   MRCLASS = {90-02 (05-02 52B55 68Q25 68R10 90C27 90C35 90C57)},
  MRNUMBER = {1956925},
MRREVIEWER = {Alexander\ I.\ Barvinok},
}

@article {NWF78,
    AUTHOR = {Nemhauser, G. L. and Wolsey, L. A. and Fisher, M. L.},
     TITLE = {An analysis of approximations for maximizing submodular set
              functions. {I}},
   JOURNAL = {Math. Programming},
  FJOURNAL = {Mathematical Programming},
    VOLUME = {14},
      YEAR = {1978},
    NUMBER = {3},
     PAGES = {265--294},
      ISSN = {0025-5610},
   MRCLASS = {90C99 (65K05)},
  MRNUMBER = {MR0503866 (58 \#20492)},
MRREVIEWER = {Jacques Dubois},
}

@article{NW78,
  title={Best algorithms for approximating the maximum of a submodular set function},
  author={Nemhauser, George L and Wolsey, Laurence A},
  journal={Mathematics of operations research},
  volume={3},
  number={3},
  pages={177--188},
  year={1978},
  publisher={INFORMS}
}

@article{Schrijver00,
	Author = {Schrijver, A.},
	Journal = {Journal of Combinatorial Theory, Series B},
	Number = {2},
	Pages = {346--355},
	Publisher = {Elsevier},
	Title = {A combinatorial algorithm minimizing submodular
		functions in strongly polynomial time},
	Volume = {80},
	Year = {2000}
}

@article{IwataFF01,
  title={A combinatorial strongly polynomial algorithm for minimizing submodular functions},
  author={Iwata, Satoru and Fleischer, Lisa and Fujishige, Satoru},
  journal={Journal of the ACM (JACM)},
  volume={48},
  number={4},
  pages={761--777},
  year={2001},
  publisher={ACM New York, NY, USA}
}

@conference{FeigeMV07,
  title={{Maximizing non-monotone submodular functions}},
  author={Feige, U. and Mirrokni, V. and Vondrak, J.},
  booktitle={Proceedings of 48th Annual IEEE Symposium on Foundations of Computer Science (FOCS)},
  year={2007}
}

@inproceedings{GoemansHIM09,
	Author = {Goemans, Michel X. and Harvey, Nicholas J. A. and
		Iwata, Satoru and Mirrokni, Vahab S.},
	Booktitle = SODA,
	Pages = {535-544},
	Title = {Approximating submodular functions everywhere},
	Year = {2009}}

@article{BuchbinderFNS15,
  author    = {Niv Buchbinder and
               Moran Feldman and
               Joseph Naor and
               Roy Schwartz},
  title     = {A Tight Linear Time (1/2)-Approximation for Unconstrained Submodular
               Maximization},
  journal   = {{SIAM} J. Comput.},
  volume    = {44},
  number    = {5},
  pages     = {1384--1402},
  year      = {2015},
  url       = {http://dx.doi.org/10.1137/130929205},
  doi       = {10.1137/130929205},
  note      = {Preliminary version in Proc.\ of IEEE FOCS 2012.}
}

@inproceedings{BalcanHI12,
    author = {Maria-Florina Balcan and Nicholas J. A. Harvey and Satoru Iwata},
    title = {Learning symmetric non-monotone submodular functions},
    booktitle = {NIPS Workshop on Discrete Optimization in Machine Learning},
    year = {2012},
    url = {https://www.cs.ubc.ca/~nickhar/papers/SymmetricSubmodular/LearningSymmetric.pdf}
    }

@inproceedings{BalkanskiS18,
  title={The adaptive complexity of maximizing a submodular function},
  author={Balkanski, Eric and Singer, Yaron},
  booktitle={Proceedings of the 50th annual ACM SIGACT symposium on theory of computing},
  pages={1138--1151},
  year={2018}
}

@inproceedings{BalkanskiRS19a,
  title={An exponential speedup in parallel running time for submodular maximization without loss in approximation},
  author={Balkanski, Eric and Rubinstein, Aviad and Singer, Yaron},
  booktitle={Proceedings of the Thirtieth Annual ACM-SIAM Symposium on Discrete Algorithms},
  pages={283--302},
  year={2019},
  organization={SIAM}
}

@inproceedings{ChenFK19,
  title={Unconstrained submodular maximization with constant adaptive complexity},
  author={Chen, Lin and Feldman, Moran and Karbasi, Amin},
  booktitle={Proceedings of the 51st Annual ACM SIGACT Symposium on Theory of Computing},
  pages={102--113},
  year={2019}
}

@article{Feige98,
  title={A threshold of ln n for approximating set cover},
  author={Feige, Uriel},
  journal={Journal of the ACM (JACM)},
  volume={45},
  number={4},
  pages={634--652},
  year={1998},
  publisher={ACM New York, NY, USA}
}

@article{GLS,
  title={The ellipsoid method and its consequences in combinatorial optimization},
  author={Gr{\"o}tschel, Martin and Lov{\'a}sz, L{\'a}szl{\'o} and Schrijver, Alexander},
  journal={Combinatorica},
  volume={1},
  number={2},
  pages={169--197},
  year={1981},
  publisher={Springer}
}

@inproceedings{li2020polynomial,
  title={A polynomial lower bound on adaptive complexity of submodular maximization},
  author={Li, Wenzheng and Liu, Paul and Vondr{\'a}k, Jan},
  booktitle={Proceedings of the 52nd Annual ACM SIGACT Symposium on Theory of Computing},
  pages={140--152},
  year={2020}
}

@inproceedings{balkanski2019optimal,
  title={An optimal approximation for submodular maximization under a matroid constraint in the adaptive complexity model},
  author={Balkanski, Eric and Rubinstein, Aviad and Singer, Yaron},
  booktitle={Proceedings of the 51st Annual ACM SIGACT Symposium on Theory of Computing},
  pages={66--77},
  year={2019}
}

@inproceedings{chekuri2019parallelizing,
  title={Parallelizing greedy for submodular set function maximization in matroids and beyond},
  author={Chekuri, Chandra and Quanrud, Kent},
  booktitle={Proceedings of the 51st Annual ACM SIGACT Symposium on Theory of Computing},
  pages={78--89},
  year={2019}
}

@inproceedings{ene2019submodular,
  title={Submodular maximization with matroid and packing constraints in parallel},
  author={Ene, Alina and Nguyen, Huy L. and Vladu, Adrian},
  booktitle={Proceedings of the 51st annual ACM SIGACT symposium on theory of computing},
  pages={90--101},
  year={2019}
}

@article {there-are-few-representable-matroids,
    AUTHOR = {Nelson, Peter},
     TITLE = {Almost all matroids are nonrepresentable},
   JOURNAL = {Bull. Lond. Math. Soc.},
  FJOURNAL = {Bulletin of the London Mathematical Society},
    VOLUME = {50},
      YEAR = {2018},
    NUMBER = {2},
     PAGES = {245--248},
      ISSN = {0024-6093,1469-2120},
   MRCLASS = {05B35},
  MRNUMBER = {3830117},
MRREVIEWER = {Eva\ Ferrara Dentice},
       DOI = {10.1112/blms.12141},
       URL = {https://doi.org/10.1112/blms.12141},
}

\appendix

\section{Relationship of submodular and XOS connectivity functions}
\label{sec:XOS-submod}

Here we prove the following.

\begin{lemma}
\label{lem:XOS-submod}
Every submodular connectivity function is also an XOS connectivity function.
\end{lemma}

\begin{proof}
Let $f:2^{[n]} \to \IR$ be a submodular connectivity function. For every permutation $\pi$ on $[n]$, define $w^\pi_{\pi(i)} = f(\pi([i])) - f(\pi([i-1]))$. By submodularity, we have $f(S) \geq \sum_{i \in S} w^\pi_i$ for every permutation $\pi$, because $f(S \cap \pi([i])) - f(S \cap \pi([i-1])) \geq f(\pi([i])) - f(\pi([i-1]))$ for $i \in S$. Also, equality is attained when $S$ is a prefix of the permutation $\pi$. Therefore, $f(S) = \max_\pi \sum_{i \in S} w^\pi_i$.

If $f(\emptyset) = f([n]) = 0$, we have $\sum_{i=1}^{n} w^\pi_i = 0$. Also, since $f$ is symmetric, 
$f(S) = f(V \setminus S) = \max_\pi \sum_{i \in V \setminus S} w^\pi_i = \max_\pi (-\sum_{i \in S} w^\pi_i)$. Hence, the formula is equally valid if we write it as $f(S) = \max_\pi |\sum_{i \in S} w^\pi_i|$.
\end{proof}

It is not true that absolute XOS functions subsume nonnegative submodular functions.
Observe that every XOS function satisfies $|f(S) - f(V \setminus S)| \leq f(\emptyset) + f(V)$ 
for all $S\subseteq V$. However, this is not satisfied by all submodular functions, 
for example the cut function of a single directed edge: $f(\emptyset) = f(\{2\}) = f(\{1,2\}) = 0$, $f(\{1\}) = 1$. 

Nevertheless, the existence of universal families for absolute XOS functions is the most general problem here: It implies the same result for XOS connectivity functions, and by Lemma~\ref{lem:XOS-submod}, we obtain the same results for submodular connectivity functions. By our discussion in Section~\ref{sec:reduction-to-submodular-connectivity-functions}, this also implies a universal family for all nonnegative submodular functions, while losing a factor of $2$.

\section{Lower bound for $k$-wise independent distributions}

Here we show a more general construction showing that even $k$-wise independent distributions cannot recover the maximum of a nonnegative submodular function within a factor better than $O(\frac{1}{\sqrt{\log n}})$.

\begin{theorem}
\label{thm:k-wise}
For every fixed $k \geq 2$ and infinitely many $n > k$ there exists a $k$-wise independent distribution $\cD$ on $2^V$ and a nonzero submodular connectivity function $f:2^V \to \IR$,
where $|V| = n$, such that

    \begin{equation*}
        \max_{H \in {\rm supp}(\cD)} f(H) = O\left( \frac{\log k}{\sqrt{\log n}} \right) \max_{A \subset V} f(A).
    \end{equation*}
\end{theorem}

\subsection{Preliminaries}

 We start with some algebraic preliminaries.

\subsubsection*{Finite fields of characteristic $2$, Frobenius maps, and trace}

Let $\F_{2^d}$ denote a finite field of size $2^d$. Its elements can be identified with elements in the vector space $V:=\F_2^d$ in such a way that the addition operation is the same in both structures.
The field structure supplies an additional multiplication operation $\F_{2^d} \times \F_{2^d}  \to \F_{2^d}$.

A known fact is that for $s\ge 0$, the Frobenius map $\Phi_s:\F_{2^d} \to \F_{2^d}$,
\[
  \Phi_s(x):=x^{2^s}
\]
is linear as a map on the vector space $\F_2^d$. Indeed, characteristic $2$ gives
$(u+v)^2=u^2+v^2$, and iteration gives $(u+v)^{2^s}=u^{2^s}+v^{2^s}$.
Also $(cu)^{2^s}=c\,u^{2^s}$ for $c\in\F_2$. Consequently, every coordinate
of $x^{2^s}$ is a linear form in the $d$ binary coordinates of $x$.
The exponent being a power of the characteristic is essential here.

The field trace on $\F_{2^d}$ is the map
\[
  \Tr:\F_{2^d} \longrightarrow\F_2,\qquad
  \Tr(z):=\sum_{s=0}^{d-1}z^{2^s}.
\]
It is linear on $\F_2^d$ by the property of the Frobenius maps above. Moreover,
\[
  (\Tr(z))^2=z^2+z^{2^2}+\cdots+z^{2^d}=\Tr(z),
\]
since the cross terms disappear due to characteristic $2$ and $z^{2^d}=z$ by the cyclic structure of the multiplicative group of $\F_{2^d}$.
Hence, the values of $\Tr(z)$ lie in $\F_2 = \{0,1\}$. 
$\Tr(z)$ is not identically zero:
it is a polynomial of  degree $2^{d-1}<|F|$ and therefore cannot vanish on all of $\F_{2^d}$.
Thus $\Tr$ is surjective, its kernel has codimension one, and for a random $Z$ in $\F_{2^d}$,
\[
  Z\text{ is uniform in } \F_{2^d}
  \quad\Longrightarrow\quad
  \Tr(Z)\text{ is a uniform bit in $\{0,1\}$}.
\]

\subsection*{Boolean polynomial spaces}

Recall that $V = \F_2^d$. We denote $n = 2^d$.
For $0\le s\le d$, let $\cP_s$ be the vector space of functions on $V$
representable by multilinear polynomials of degree at most $s$:
\[
  \cP_s:=
  \left\{  g:V \to \F_2; \
    g(x)=\sum_{\substack{T\subset[d]\\|T|\le s}}
       c_T\prod_{i\in T}x_i, \ c_T\in\F_2
  \right\}.
\]
We identify such a polynomial with its vector of $2^d$ evaluations $(g(x): x \in V)$,
or with its vector of coefficients $(c_T: T \subset [d], |T| \leq s)$.
By the second representation, and linear independence of the monomials in $\cP_s$ over $\IF_2$, we have
\[
  \dim\cP_s=\sum_{i=0}^s\binom di.
\]
For Boolean functions $q,g:V \to \F_2$, we define pointwise multiplication:
$(qg)(x)=q(x)g(x)$. Hence
\[
  q\in\cP_r,\quad g\in\cP_s
  \quad\Longrightarrow\quad
  qg\in\cP_{r+s},
\]
after multilinearization using $x_i^2=x_i$.

\subsection{A $k$-wise independent distribution of small support}

Choose independent uniform coefficients
$a_0,\ldots,a_{k-1}\in \F_{2^d}$ and define
\[
  P_a(x):=\sum_{j=0}^{k-1}a_jx^j,\qquad
  q_a(x):=\Tr(P_a(x)),\qquad
  H_a:=\{x\in \F_{2^d}:q_a(x)=1\}.
\]
The random coefficient tuple $a$ induces a distribution $\mathcal D$ on
subsets $H_a\subset V$.

\begin{theorem}
The membership variables
$\bigl(\mathbf 1[x\in H_a]\bigr)_{x\in V}$ are unbiased and $k$-wise
independent. The distribution has support at most $n^k$ and can be sampled
with $kd$ random bits.
\end{theorem}

\begin{proof}
Fix distinct $x_1,\ldots,x_t\in \F_{2^d}$, where $t\le k$. The evaluation map
\[
  (a_0,\ldots,a_{k-1})
  \longmapsto
  (P_a(x_1),\ldots,P_a(x_t))
\]
has a $t$-row Vandermonde matrix. Its first $t$ columns have determinant
$\prod_{i<j}(x_j-x_i)\ne0$, so the map has rank $t$. Therefore the
$t$ field values are independent and uniform in $\F_{2^d}$. Applying the
nonzero linear map $\Tr$ coordinatewise gives $t$ independent uniform
bits. Finally, there are $|\F_{2^d}|^k=n^k$ coefficient tuples,
and $\log_2 |\F_{2^d}|^k = k \log_2 n = kd$.
\end{proof}

\subsection{Boolean degree of the sampled polynomials}

It might seem at first sight that the construction above leads to polynomials
of degree $k$. However, with some care we will see that the degree
important for us is actually $\log k$.
Fix the coefficients $a_j$ and view $q_a$ as a function of the $d$ binary
coordinates of $x= (x_1,\ldots,x_d)$. If
$j=\sum_s\epsilon_s2^s$, $\epsilon_s\in\{0,1\}$
then
\[
  x^j=\prod_{s:\epsilon_s=1}x^{2^s}.
\]
Each Frobenius factor $x^{2^s}$
is linear as a map on the vector space $\F_2^d$, and
field multiplication is bilinear on $\F_2^d$. Thus the degree of $x^j$
as a polynomial in $(x_1,\ldots,x_d)$ over $\F_2$ is
\[
  \deg_{\F_2}(x^j)\le \sum_s \epsilon_s \leq \lceil \log_2 k \rceil.
\]
Multiplication by $a_j$, application of $\Tr$, and summation over $j$ are
all linear operations on $\F_2^d$, so
\[
  q_a\in\cP_{r_k},
  \qquad
  r_k := \lceil \log_2 k \rceil.
\]
This is multilinear Boolean degree in the variables $x_1,\ldots,x_d \in \F_2$;
observe that the ordinary univariate degree of $P_a$ as a polynomial over $\F_{2^d}$ is $k-1$ with probability $1-1/2^d$ (whenever $a_{k-1} \neq 0$).

\subsection{The submodular lower-bound construction}

Assume $d=2m+1$ and $r_k\le m$, as holds for every fixed $k$ once $d$
is sufficiently large. Let $M$ be the binary matrix whose rows indexed by $T$
are the evaluation vectors of each monomial $\prod_{i \in T} x_i$ of degree
$|T| \leq m$ on $V$. Its row space is $\cP_m$, and
\[
  R:=\operatorname{rank}(M)=\dim\cP_m
    =\sum_{i=0}^m\binom di
    =2^{d-1}=\frac n2.
\]
For $S\subset V$, let $\rho(S)$ be the rank of the columns of $M$
indexed by $S$, and define
\[
  f(S):=\rho(S)+\rho(V\setminus S)-\rho(V).
\]
This is a submodular connectivity function:
The function is symmetric by definition. It is nonnegative because
$\rho(V)\le\rho(S)+\rho(V\setminus S)$, it is submodular because
both $S\mapsto\rho(S)$ and $S\mapsto\rho(V\setminus S)$ are submodular, and it satisfies
$f(\emptyset) = f(V) = \rho(\emptyset) + \rho(V) - \rho(V) = 0$. Also, recall that $\rho(V) = \dim \cP_m = R$.

The rows of $M$ are mutually orthogonal under
$\langle u,v\rangle=\sum_{x\in V}u(x)v(x)$ over $\F_2$. Indeed, the
product of two monomials of degree at most $m$ has degree at most
$2m=d-1$, so some coordinate is absent; summing over that coordinate
gives zero in $\F_2$. Hence $M M^T = 0$. Since $\dim\cP_m=n/2$, the space $\cP_m$ equals
its orthogonal complement.

We claim that there exist two sets $A, V \setminus A$ such that $\rho(S) = \rho(V \setminus S) = R = \frac{n}{2}$. The rank of $M$ is $R = n/2$. Let $A$ be any subset of $R$ columns which are linearly independent. The matrix $M_A$ restricted to columns indexed by $A$ is invertible, and
we have $\det(M_A M_A^T) = (\det (M_A))^2 \neq 0$. On the other hand, we have
$M_A M_A^T + M_{V \setminus A} M_{V \setminus A}^T = M M^T = 0$. Hence
$(\det (M_{V \setminus A}))^2 = \det(M_{V \setminus A} M_{V \setminus A}^T) = \det(M_A M_A^T) \neq 0$
and $M_{V \setminus A}$ is also invertible, which means $\rho(A) = \rho(V \setminus A) = R$.
To summarize,
\[
  f(A) = \rho(A) + \rho(V \setminus A) - R =  R = \frac n2.
\]

On the other hand, every set $A \subset V$ satisfies $f(A) \leq R + R - R = R = \frac{n}{2}$.
Hence, $\max_{A \subset V} f(A) = \frac{n}{2}$.

Now fix any sampled function $q=q_a\in\cP_{r_k}$ and let
$H=\supp(q)=\{x:q(x)=1\}$. For $S\subset V$, define
\[
  W_S:=\{h\in\cP_m:\supp(h)\subset S\}.
\]
For every $g\in\cP_{m-r_k}$, set
\[
  h_1(x) := q(x) g(x),\qquad h_0(x) := (1+q(x)) g(x).
\]
Then $h_1,h_0\in\cP_m$, with
$\supp(h_1)\subset H$ and
$\supp(h_0)\subset V\setminus H$. Moreover,
$h_1+h_0=g$, so
\[
  g\longmapsto(h_1,h_0)
\]
is an injective linear map from $\cP_{m-r_k}$ into
$W_H\oplus W_{V\setminus H}$. Hence
\[
  \dim W_H+\dim W_{V\setminus H}\ge\dim\cP_{m-r_k}.
\]
Restriction of functions in $\cP_m$ to $S$ has rank $\rho(S)$ and kernel
$W_{V\setminus S}$. Therefore
\[
  \rho(S)=\dim\cP_m-\dim W_{V\setminus S}.
\]
Substituting this identity into the definition of $f$ gives
\begin{align*}
  f(H)
   &=\dim\cP_m-\dim W_H-\dim W_{V\setminus H}\\
   &\le\dim\cP_m-\dim\cP_{m-r_k}\\
   &=\sum_{i=m-r_k+1}^{m}\binom di
    =O\!\left(\frac{r_k 2^d}{\sqrt d}\right).
\end{align*}
The last estimate uses the standard bound on binomial coefficients,
$\binom{d}{i}=O(2^d/\sqrt d)$ for all $i$.

Since $d=\log_2n$ and $r_k=O(\log k)$, every set in the support of
$\mathcal D$ satisfies
\[
  f(H)=O\!\left(\frac{n\log k}{\sqrt{\log n}}\right).
\]
Together with $\max_Af(A)=n/2$, this yields
\[
  \frac{\max_{A\subset V}f(A)}
       {\max_{H\in\supp(\mathcal D)}f(H)}
  =\Omega\!\left(\frac{\sqrt{\log n}}{\log k}\right).
\]
For fixed $k$, the denominator $\log k$ is constant, proving the
$\Omega(\sqrt{\log n})$ lower bound.

\end{document}